\documentclass{llncs}
\usepackage{amssymb, color}
\usepackage{pifont}
\usepackage{graphicx}
\usepackage[justification=centering]{caption}
\usepackage{subfigure}
\usepackage{multirow, multicol}
\usepackage{booktabs, ctable}
\usepackage{algorithm}
\usepackage{algorithmicx}
\usepackage{algpseudocode}
\usepackage{comment}

\usepackage{fancybox}
\newcommand{\cmark}{\ding{51}}
\newcommand{\xmark}{\ding{55}}

\begin{document}

\title{Dynamic Searchable Symmetric Encryption Schemes Supporting Range Queries with Forward/Backward Privacy}

\author{Cong Zuo\inst{1,2} \and
	Shi-Feng Sun\inst{1,2,} \and
	Joseph K. Liu\inst{1} \and
	Jun Shao \inst{3}\and 
	Josef Pieprzyk\inst{2,4}}

\institute{Faculty of Information Technology, Monash University, Clayton, 3168, Australia \\
	\email{\{cong.zuo1,shifeng.sun,joseph.liu\}@monash.edu}
	 \and Data61, CSIRO, Melbourne/Sydney, Australia \\
	 \email{ josef.pieprzyk@data61.csiro.au}
	\and School of Computer and Information Engineering, Zhejiang Gongshang University, Hangzhou 310018, Zhejiang, China\\
	\email{chn.junshao@gmail.com}
	\and Institute of Computer Science, Polish Academy of Sciences, 01-248 Warsaw,  Poland
}

\maketitle

\begin{abstract}
Dynamic searchable symmetric encryption (DSSE) is a useful cryptographic tool in encrypted cloud storage. However, it has been reported that DSSE usually suffers from file-injection attacks and content leak of deleted documents. To mitigate these attacks, forward privacy and backward privacy have been proposed. Nevertheless, the existing forward/backward-private DSSE schemes can only support single keyword queries. To address this problem, in this paper, we propose two DSSE schemes supporting range queries. One is forward-private and supports a large number of documents. The other can achieve backward privacy, while it can only support a limited number of documents. Finally, we also give the security proofs of the proposed DSSE schemes in the random oracle model.
\keywords{Dynamic searchable symmetric encryption \and Forward privacy \and Backward privacy \and Range queries}
\end{abstract}

\section{Introduction}
Searchable symmetric encryption (SSE) is a useful cryptographic primitive that can encrypt the data to protect its confidentiality while keeping its searchability. Dynamic SSE (DSSE) further provides data dynamics that allows the client to update data over the time without losing data confidentiality and searchability. Due to this property, DSSE is highly demanded in encrypted cloud. However, many existing DSSE schemes \cite{KPR12,CJJJKRS14} suffer from file-injection attacks \cite{CGPR15,ZKP16}, where the adversary can compromise the privacy of a client query by injecting a small portion of new documents to the encrypted database. To resist this attack, Zhang et al. \cite{ZKP16} highlighted the need of forward privacy that was informally introduced by Stefanov et al. \cite{SPS14}. The formal definition of forward privacy for DSSE was given by Bost \cite{Bos16} who also proposed a concrete forward-private DSSE scheme. Furthermore, Bost et al. \cite{BMO17} demonstrated the damage of content leak of deleted documents and proposed the corresponding security notion—backward privacy. Several backward-private DSSE schemes were also presented in \cite{BMO17}. 

Nevertheless, the existing forward/backward-private DSSE schemes only support single keyword queries, which are not expressive enough in data search service \cite{FJKNRS15,DPPDG16}. To solve this problem, we aim to design forward/backward-private DSSE schemes supporting range queries. Our design starts from the regular binary tree in \cite{FJKNRS15} to support range queries. However, the binary tree in \cite{FJKNRS15} cannot be applied directly to the dynamic setting. It is mainly because that the keywords in \cite{FJKNRS15} are labelled according to the corresponding tree levels that will change significantly in the dynamic setting. A na\"{i}ve solution is to replace all old keywords by the associated new keywords. This is, however, not efficient. To address this problem, we have to explore new approaches for our goal.

\noindent \textbf{Our Contributions.} To achieve the above goal, we propose two new DSSE constructions supporting range queries in this paper. The first one is forward-private but with a larger client overhead in contrast to \cite{FJKNRS15}. The second one is backward-private DSSE which greatly reduces the client and the server storage at the cost of losing forward privacy. In more details, our main contributions are as follows:

\begin{itemize}
\item To make the binary tree suitable for range queries in the dynamic setting, we introduce a new binary tree data structure, and then present the first forward-private DSSE supporting range queries by applying it to Bost's scheme \cite{Bos16}. However, the forward privacy is achieved at the expense of suffering from a large storage overhead on the client side.
\item To achieve backward privacy, we apply the Paillier cryptosystem and the bit string representation to Bost's framework. To reduce the storage at both the client and server side, we used a fixed update token for each keyword. Note that, this scheme is not forward-private, since, for every update, the server knows which keyword has been updated. We refer readers to Section \ref{sec:construction} for more details. Notably, due to the limitation of the Paillier cryptosystem, it cannot support large-scale database consisting a large number of documents. Nevertheless, it suits well for certain scenarios where the number of documents is moderate. The new approach may give new lights on designing more efficient and secure DSSE schemes.
\item Also, the comparison with related works in Table \ref{tab:comparison} and detailed security analyses are provided, which demonstrate that our constructions are not only forward/backward-private but also with a comparable efficiency.
\end{itemize}

%\noindent\textit{Remark:} In the conference version \cite{ZSLSP18}, we claim that our second scheme can achieve forward privacy. However, it cannot achieve forward privacy, since we remove the dictionary at the client side to save storage. In the encrypted database, every keyword shares the same update token (address). Hence, for each update, the server knows which keyword has been updated. Note that, if we want to achieve forward privacy, we can use the same framework \cite{Bos16} as our first construction at the cost of large client storage. Nevertheless, we failed to change this in the conference version.
\noindent\textit{Remark:} In this full version, we correct the wrong theorem  (cf. Theorem 2) in the conference version \cite{ZSLSP18}, where we claimed that the second construction can achieve forward privacy. In fact, the forward privacy cannot be achieved, because the update can be linked to previous searches by the fixed token. In particular, we focus on achieving backward privacy by applying the Paillier cryptosystem and the bit string representation to Bost \cite{Bos16}'s framework, which can also achieve forward privacy due to Bost's technique.  For every update on keyword $n$, however, the server needs to store a ciphertext  (of Paillier encryption)  corresponding to $n$.  In addition, the number of keywords is nearly doubled compared to the scheme of \cite{Bos16}, as mentioned in our first construction.  Then this construction incurs large storage on both the server and the client side. In the conference version, we used a fixed update token for each keyword to further reduce the storage overhead. By this way, the server can homomorphicly add the ciphertext to the previous ones corresponding to the same keyword, and the client does not need to store the current search token yet. Therefore, both the client and the server storage are reduced a lot, but at the cost of losing forward privacy. Nevertheless, we made some mistake when preparing the conference version and so we correct it here.

\begin{table}[!htb]
	\centering
	\caption{Comparison with existing DSSE schemes}\label{tab:comparison}
	\begin{tabular}{|c|c|c|c|c|c|c|c|}
		\hline
		\multirow{2}{*}{Scheme} & \multicolumn{2}{c|}{Client Computation} & Client & Range & Forward & Backward & Document \\
		\cline{2-3}
		& Search & Update & Storage & Queries & Privacy & Privacy & number\\
		\hline
		\cite{FJKNRS15} & $w_R$ & - & $O(1)$ & \cmark & \xmark & \xmark & large\\
		\hline
		\cite{Bos16} & - & $O(1)$ & $O(W)$ &  \xmark & \cmark & \xmark & large \\
		\hline
		Ours A & $w_R$ & $\lceil log(W)\rceil + 1$ & $O(2W)$ & \cmark & \cmark & \xmark & large\\
		\hline
		Ours B & $w_R$ & $\lceil log(W)\rceil + 1$ & $O(1)$ & \cmark & \xmark & \cmark & small\\
		\hline
		\multicolumn{8}{p{11cm}<{}}{$W$ is the number of keywords in a database, $w_R$ is the number of keywords for a range query(we map a range query to a few different keywords).}
	\end{tabular}
\end{table}

\subsection{Related Works}
Song et al. \cite{SWP00} were the first using symmetric encryption to facilitate keyword search over the encrypted data. Later, Curtmola et al. \cite{CGKO06} gave a formal definition for SSE and the corresponding security model in the static setting. To make SSE more scalable and expressive, Cash et al. \cite{CJJKRS13} proposed a new scalable SSE supporting Boolean queries. Following this construction, many extensions have been proposed. Faber et al. \cite{FJKNRS15} extended it to process a much richer collection of queries. For instance, they used a binary tree with keywords labelled according to the tree levels to support range queries. Zuo et al. \cite{ZMYSL16} made another extension to support general Boolean queries. Cash et al.'s construction has also been extended into multi-user setting \cite{SLSSY16,KLS17,WWSLSC17}. However, the above schemes cannot support data update. To solve this problem, some DSSE schemes have been proposed \cite{KPR12,CJJJKRS14}.

However, designing a secure DSSE scheme is not an easy job. Cash et al. \cite{CGPR15} pointed out that only a small leakage leveraged by the adversary would be enough to compromise the privacy of clients' queries. A concrete attack named file-injection attack was proposed by Zhang et al. \cite{ZKP16}. In this attack, the adversary can infer the concept of aclient queries by injecting a small portion of new documents into encrypted database. This attack also highlights the need for forward privacy which protects security of new added parts. Accordingly, we have backward privacy that protects security of new added parts and later deleted. These two security notions were first introduced by Stefanov et al.\cite{SPS14}. The formal definitions of forward/backward privacy for DSSE were given by Bost \cite{Bos16} and Bost et al. \cite{BMO17}, respectively. In \cite{Bos16}, Bost also proposed a concrete forward-private DSSE scheme, it does not support physical deletion. Later on, Kim et al. \cite{KKLPK17} proposed a forward-private DSSE scheme supporting physical deletion. Meanwhile, Bost et al. \cite{BMO17} proposed a forward/backward-private DSSE to reduce leakage during deletion. Unfortunately, all the existing forward/backward-private DSSE schemes only support single keyword queries. Hence, forward/backward-private DSSE supporting more expressive queries, such as range queries, are quite desired.

Apart from the binary tree technique, order preserving encryption (OPE) can also be used to support range queries. The concept of OPE was proposed by Agrawal et al. \cite{AKSX04}, and it allows the order of the plaintexts to be preserved in the ciphertexts. It is easy to see that this kind of encryption would lead to the leakage in \cite{BCLO09,BCO11}. To reduce this leakage, Boneh et al. \cite{BLRSZZ15} proposed another concept named order revealing encryption (ORE), where the order of the ciphertexts are revealed by using an algorithm rather than comparing the ciphertexts (in OPE) directly. More efficient ORE schemes were proposed later \cite{CLWW16}. However, ORE-based SSE still leaks much information about the underlying plaintexts. To avoid this, in this paper, we focus on how to use the binary tree structure to achieve range queries. 

\subsection{Organization}
The remaining sections of this paper are organized as follows. In Sect. \ref{sec:pre}, we give the background information and building blocks that are used in this paper. In Sect. \ref{sec:sse}, we give the definition of DSSE and its security definition. After that in Sect. \ref{sec:construction}, we present a new binary tree and our DSSE schemes. Their security analyses are given in Sect. \ref{sec:security}. Finally, Sect. \ref{sec:conclusion} concludes this work.

\section{Preliminaries}\label{sec:pre}
In this section, we describe cryptographic primitives (building blocks) that are used in this work.
\subsection{Trapdoor Permutations}
A trapdoor permutation (TDP) $\Pi$ is a one-way permutation over a domain $D$ such that (1) it is ``easy'' to compute $\Pi$ for any value of the domain with the public key, and (2) it is ``easy'' to calculate the inverse $\Pi^{-1}$ for any value of a co-domain $\mathcal{M}$ only if a matching secret key is known. More formally, $\Pi$ consists of the following algorithms:
\begin{itemize}
	\item $\texttt{TKeyGen}(1^{\lambda})\rightarrow (\texttt{TPK,\texttt{TSK}})$: For a security parameter $1^{\lambda}$, the algorithm returns a pair of cryptographic keys: a public key $\texttt{TPK}$ and a secret key $\texttt{TSK}$.
	
	\item $\Pi(\texttt{TPK},x)\rightarrow y$: For a pair: public key $\texttt{TPK}$ and $x\in D$, the algorithm outputs $y\in \mathcal{M}$.
	
	\item $\Pi^{-1}(\texttt{TSK},y)\rightarrow x$: For a pair: a secret key $\texttt{TSK}$ and $y\in\mathcal{M}$, the algorithm returns $x\in D$.
\end{itemize}

\noindent\textbf{One-wayness.} We say $\Pi$ is one-way if for any probabilistic polynomial time (PPT) adversary $\mathcal{A}$, an advantage \[\verb"Adv"_{\Pi,\mathcal{A}}^{\verb"OW"}(1^{\lambda})=\Pr[ x\leftarrow \mathcal{A}(\texttt{TPK}, y)]\] is negligible, where $(\texttt{TSK},\texttt{TPK})\leftarrow \texttt{TKeyGen}(1^{\lambda})$, $y\leftarrow\Pi(\texttt{TPK}, x)$, $x\in D$.

\subsection{Paillier Cryptosystem}
A Paillier cryptosystem $\Sigma=(\texttt{KeyGen}, \texttt{Enc}, \texttt{Dec})$ is defined by following three algorithms:
\begin{itemize}
	\item $\texttt{KeyGen}(1^{\lambda})\rightarrow (\texttt{PK,SK})$: It chooses at random two primes $p$ and $q$ of similar lengths and computes $n=pq$ and $\phi(n)=(p-1)(q-1)$. Next it sets $g=n+1$, $\beta=\phi(n)$ and $\mu=\phi(n)^{-1}$ mod $n$. It returns $\texttt{PK}=(n,g)$ and $\texttt{SK}=(\beta, \mu)$.
	
	\item $\texttt{Enc}(\texttt{PK},m)\rightarrow c$: Let $m$ be the message, where $0\le m<n$, the algorithm selects an integer $r$ at random from $\mathbb{Z}_n$ and computes a ciphertext $c=g^m\cdot r^n$ mod $n^2$.
	
	\item $\texttt{Dec}(\texttt{SK},c)\rightarrow m$: The algorithm calculates $m=L(c^{\beta}$ mod $n^2)\cdot \mu$ mod $n$, where $L(x)=\frac{x-1}{n}$.
\end{itemize}

\noindent\textbf{Semantically Security.} We say $\Sigma$ is semantically secure if for any probabilistic polynomial time (PPT) adversary $\mathcal{A}$, an advantage \[\verb"Adv"_{\Sigma,\mathcal{A}}^{\verb"IND-CPA"}(1^{\lambda})=|\Pr[\mathcal{A}(\texttt{Enc}(\texttt{PK}, m_0))=1]-\Pr[\mathcal{A}(\texttt{Enc}(\texttt{PK}, m_1))=1]|\] is negligible, where $(\texttt{SK}$, $\texttt{PK})\leftarrow \texttt{KeyGen}(1^{\lambda})$, $\mathcal{A}$ chooses $m_0$, $m_1$ and $|m_0|=|m_1|$.

\noindent\textbf{Homomorphic Addition.} Paillier cryptosystem is homomorphic, i.e. \[\texttt{Dec}(\texttt{Enc}(m_1)\cdot \texttt{Enc}(m_2))\bmod{n^2}=m_1+m_2\bmod{n}.\] In our second construction, we need this property to achieve backward privacy.

\subsection{Notations}
The list of notations used is given in Table \ref{tab:notation}.

\begin{table}[!htb]
	\centering
	\caption{Notations (used in our constructions)}\label{tab:notation}
	\begin{tabular}{|c|p{11cm}<{}|}
		\hline
		$W$ & The number of keywords in a database \texttt{DB} \\
		\hline
		\texttt{BDB} & The binary database which is constructed from a database \texttt{DB} by using our binary tree \texttt{BT} \\
		\hline
		$m$ & The number of values in the range $[0, m-1]$ for our range queries \\
		\hline
		$v$ & A value in the range $[0, m-1]$ where $0\le v<m$ \\
		\hline
		$n_i$ & The $i$-th node in our binary tree which is considered as the keyword \\
		\hline
		$\texttt{root}_{o}$ & The root node of the binary tree before update \\
		\hline
		$\texttt{root}_n$ & The root node of the binary tree after update \\
		\hline
		$ST_c$ & The current search token for a node $n$ \\
		\hline
		$\mathcal{M}$ & A random value for $ST_0$ which is the first search token for a node $n$ \\
		\hline
		$UT_c$ & The current update token for a node $n$ \\
		\hline
		\textbf{T} & A map which is used to store the encrypted database \texttt{EDB} \\
		\hline
		\textbf{N} & A map which is used to store the current search token for $n_i$ \\
		\hline
		\textbf{NSet} & The node set which contains the nodes \\
		\hline
		$\texttt{TPK}$ & The public key of trapdoor permutation \\
		\hline
		$\texttt{TSK}$ & The secret key of trapdoor permutation \\
		\hline
		$\texttt{PK}$ & The public key of Paillier cryptosystem \\
		\hline
		$\texttt{SK}$ & The secret key of Paillier cryptosystem \\
		\hline
		$f_i$ & The $i$-th file \\
		\hline
		\texttt{PBT} & Perfect binary tree \\
		\hline
		\texttt{CBT} & Complete binary tree \\
		\hline
		\texttt{VBT} & Virtual perfect binary tree \\
		\hline
		\texttt{ABT} & Assigned complete binary tree \\
		\hline
	\end{tabular}
\end{table}

\section{Dynamic Searchable Symmetric Encryption (DSSE)}\label{sec:sse}
We follow the database model given in the paper \cite{Bos16}. A database is a collection of (index, keyword set) pairs denoted as \texttt{DB}$=(ind_i,\textbf{W}_i)_{i=1}^d$, where $ind_i\in\{0,1\}^{\ell}$ and $\textbf{W}_i\subseteq\{0,1\}^*$. The set of all keywords of the database \texttt{DB} is $\textbf{W}=\cup_{i=1}^d\textbf{W}_i$, where $d$ is the number of documents in \texttt{DB}. We identify $W=|\textbf{W}|$ as the total number of keywords and $N=\Sigma_{i=1}^d|\textbf{W}_i|$ as the number of document/keyword pairs. We denote \texttt{DB}($w$) as the set of documents that contain a keyword $w$. To achieve a sublinear search time, we encrypt the file indices of \texttt{DB}($w$) corresponding to the same keyword $w$ (a.k.a. inverted index\footnote{It is an index data structure where a word is mapped to a set of documents which contain this word.}).

A DSSE scheme $\Gamma$ consists of an algorithm \textbf{Setup} and two protocols \textbf{Search} and \textbf{Update} as described below.

\begin{itemize}
\item (\texttt{EDB}, $\sigma$) $\leftarrow$ \textbf{Setup}(\texttt{DB}, $1^{\lambda}$): For a security parameter $1^{\lambda}$ and a database \texttt{DB}. The algorithm outputs an encrypted database \texttt{EDB} for the server and a secret state $\sigma$ for the client.
\item ($\mathcal{I}$, $\perp$) $\leftarrow$ \textbf{Search}($q$, $\sigma$, \texttt{EDB}): The protocol is executed between a client (with her query $q$ and state $\sigma$) and a server (with its \texttt{EDB}). At the end of the protocol, the client outputs a set of file indices $\mathcal{I}$ and the server outputs nothing.
\item ($\sigma'$, \texttt{EDB}$'$) $\leftarrow$ \textbf{Update}($\sigma$, $op$, $in$, \texttt{EDB}): The protocol runs between a client and a server. The client input is a state $\sigma$, an operation $op=(add,del)$ she wants to perform and a collection of $in=(ind, \textbf{w})$ pairs that are going to be modified, where $add,del$ mean the addition and deletion of a document/keyword pair, respectively, $ind$ is the file index and \textbf{w} is a set of keywords. The server input is \texttt{EDB}. \textbf{Update} returns an updated state $\sigma'$ to the client and an updated encrypted database \texttt{EDB}$'$ to the server.
\end{itemize}

\subsection{Security Definition}\label{subsec:security}
The security definition of DSSE is formulated using the following two games: $\verb"DSSEREAL"_{\mathcal{A}}^{\Gamma}(1^{\lambda})$ and $\verb"DSSEIDEAL"_{\mathcal{A},\mathcal{S}}^{\Gamma}(1^{\lambda})$. The $\verb"DSSEREAL"_{\mathcal{A}}^{\Gamma}(1^{\lambda})$ is executed using DSSE. The $\verb"DSSEIDEAL"_{\mathcal{A},\mathcal{S}}^{\Gamma}(1^{\lambda})$ is simulated using the leakage of DSSE. The leakage is parameterized by a function $\mathcal{L}=(\mathcal{L}^{Stp}, \mathcal{L}^{Srch}, \mathcal{L}^{Updt})$, which describes what information is leaked to the adversary $\mathcal{A}$. If the adversary $\mathcal{A}$ cannot distinguish these two games, then we can say there is no other information leaked except the information that can be inferred from the leakage function $\mathcal{L}$. More formally,
\begin{itemize}
    \item $\verb"DSSEREAL"_{\mathcal{A}}^{\Gamma}(1^{\lambda})$: On input a database \texttt{DB}, which is chosen by the adversary $\mathcal{A}$, it outputs \texttt{EDB} by using $\textbf{Setup}(1^{\lambda},$ \texttt{DB}) to the adversary $\mathcal{A}$. $\mathcal{A}$ can repeatedly perform a search query $q$ (or an update query $(op,in$)). The game outputs the results generated by running \textbf{Search}($q$) (or \textbf{Update}($op,in$)) to the adversary $\mathcal{A}$. Eventually, $\mathcal{A}$ outputs a bit.
    \item $\verb"DSSEIDEAL"_{\mathcal{A},\mathcal{S}}^{\Gamma}(1^{\lambda})$: On input a database \texttt{DB} which is chosen by the adversary $\mathcal{A}$, it outputs \texttt{EDB} to the adversary $\mathcal{A}$ by using a simulator $\mathcal{S}(\mathcal{L}^{Stp}(1^{\lambda}$, \texttt{DB})). Then, it simulates the results for the search query $q$ by using the leakage function $\mathcal{S}(\mathcal{L}^{Srch}(q))$ and uses $\mathcal{S}(\mathcal{L}^{Updt}(op,in))$ to simulate the results for update query ($op,in$). Eventually, $\mathcal{A}$ outputs a bit.
\end{itemize}

\begin{definition}
A DSSE scheme $\Gamma$ is $\mathcal{L}$-adaptively-secure if for every PPT adversary $\mathcal{A}$, there exists an efficient simulator $\mathcal{S}$ such that \[|\Pr[\verb"DSSEREAL"_{\mathcal{A}}^{\Gamma}(1^{\lambda})=1]-\Pr[\verb"DSSEIDEAL"_{\mathcal{A},\mathcal{S}}^{\Gamma}(1^{\lambda})=1]|\le negl(1^{\lambda}).\]
\end{definition}

\section{Constructions}\label{sec:construction}
In this section, we give two DSSE constructions. In order to process range queries, we deploy a new binary tree which is modified from the binary tree in \cite{FJKNRS15}. Now, we first give our binary tree used in our constructions. 

\subsection{Binary Tree for Range Queries}\label{subsec:binary-tree}
In a binary tree \texttt{BT}, every node has at most two children named $left$ and $right$. If a node has a child, then there is an edge that connects these two nodes. The node is the parent $parent$ of its child. The root $root$ of a binary tree does not have parent and the leaf of a binary tree does not have any child. In this paper, the binary tree is stored in thew form of linked structures. The first node of \texttt{BT} is the root of a binary tree. For example, the root node of the binary tree \texttt{BT} is \texttt{BT}, the left child of \texttt{BT} is \texttt{BT}.$left$, and the parent of \texttt{BT}'s left child is \texttt{BT}.$left$.$parent$, where \texttt{BT} = \texttt{BT}.$left$.$parent$.

In a complete binary tree \texttt{CBT}, every level, except possibly the last, is completely filled, and all nodes in the last level are as far left as possible (the leaf level may not full). A perfect binary tree \texttt{PBT} is a binary tree in which all internal nodes (not the leaves) have two children and all leaves have the same depth or same level.  Note that, \texttt{PBT} is a special \texttt{CBT}.

\subsection{Binary Database} In this paper, we use binary database \texttt{BDB} which is generated from \texttt{DB}. In \texttt{DB}, keywords (the first row in \ref{fig:BT}.(c)) are used to retrieve the file indices (every column in \ref{fig:BT}.(c)). For simplicity, we map keywords in \texttt{DB} to the values in the range [0,$m-1$] for range queries\footnote{In different applications, we can choose different kinds of values. For instance, audit documents of websites with particular IP addresses. We can search the whole network domain, particular host or application range.}, where $m$ is the maximum number of values. If we want to search the range [0,3], a na\"{i}ve solution is to send every value in the range (0, 1, 2 and 3) to the server, which is not efficient. To reduce the number of keywords sent to the server, we use the binary tree as shown in Fig. \ref{fig:BT}.(a). For the range query [0,3], we simply send the keyword $n_3$ (the minimum nodes to cover value 0, 1, 2 and 3) to the server. In \texttt{BDB}, every node in the binary tree is the keyword of the binary database, and every node has all the file indices for its decedents, as illustrated in Figure \ref{fig:BT}.(d).

As shown in Fig. \ref{fig:BT}.(a), keyword in \texttt{BDB} corresponding to node $i$ (the black integer) is $n_i$ (e.g. the keyword for node 0 is $n_0$.). The blue integers are the keywords in \texttt{DB} and are mapped to the values in the range [0,3]. These values are associated with the leaves of our binary tree. The words in red are the file indices in \texttt{DB}. For every node (keyword), it contains all the file indices in its descendant leaves. Node $n_1$ contains $f_0,f_1,f_2,f_3$ and there is no file in node $n_4$ (See Fig. \ref{fig:BT}.(d)). For a range query $[0,2]$, we need to send the keywords $n_1, n_4$ ($n_1$ and $n_4$ are the minimum number of keywords to cover the range $[0,2]$.) to the server, and the result file indices are $f_0, f_1, f_2$ and $f_3$.

\subsubsection{Bit String Representation.} \label{subsubsec:bs}

We parse the file indices for every keyword in \texttt{BDB} (e.g. every column in Figure \ref{fig:BT}.(d)) into a bit string, which we will use later. Suppose there are $y-1$ documents in our \texttt{BDB}, then we need $y$ bits to represent the existence of these documents. The highest bit is the sign bit (0 means positive and 1 means negative). If $f_i$ contains keyword $n_j$, then the $i$-th bit of the bit string for $n_j$ (every keyword has a bit string) is set to 1. Otherwise, it is set to 0. For update, if we want to add a new file index $f_i$ (which also contains keyword $n_j$) to keyword $n_j$, we need a positive bit string, where the $i$-th bit is set to 1 and all other bits are set to 0. Next, we add this bit string to the existing bit string associated with $n_j$ \footnote{Note that, in the range queries, the bit strings are bit exclusive since a file is corresponded to one value only.}. Then, $f_i$ is added to the bit string for $n_j$. If we want to delete file index $f_i$ from the bit string for $n_j$, we need a negative bit string (the most significant bit is set to 1), the $i$-th bit is set to 1 and the remaining bits are set to 0. Then, we need to get the complement of the bit string \footnote{In a computer, the subtraction is achieved by adding the complement of the negative bit string.}. Next, we add the complement bit string as in the add operation. Finally, the $f_i$ is deleted from the bit string for $n_j$.

For example, in Fig. \ref{fig:BT}.(b), the bit string for $n_0$ is 000001, and the bit string for $n_4$ is 000000. Assume that we want to delete file index $f_0$ from $n_0$ and add it to $n_4$. First we need to generate bit string 000001 and add it to the bit string (000000) for $n_4$. Next we generate the complement bit string 111111 (the complement of 100001) and add it to 000001 for $n_0$. Then, the result bit strings for $n_0$ and $n_4$ are 000000 and 000001, respectively. As a result, the file index $f_0$ has been moved from $n_0$ to $n_4$.

\begin{figure}[htb]
	\centering
	\includegraphics[width=\textwidth,trim=0 20 0 50]{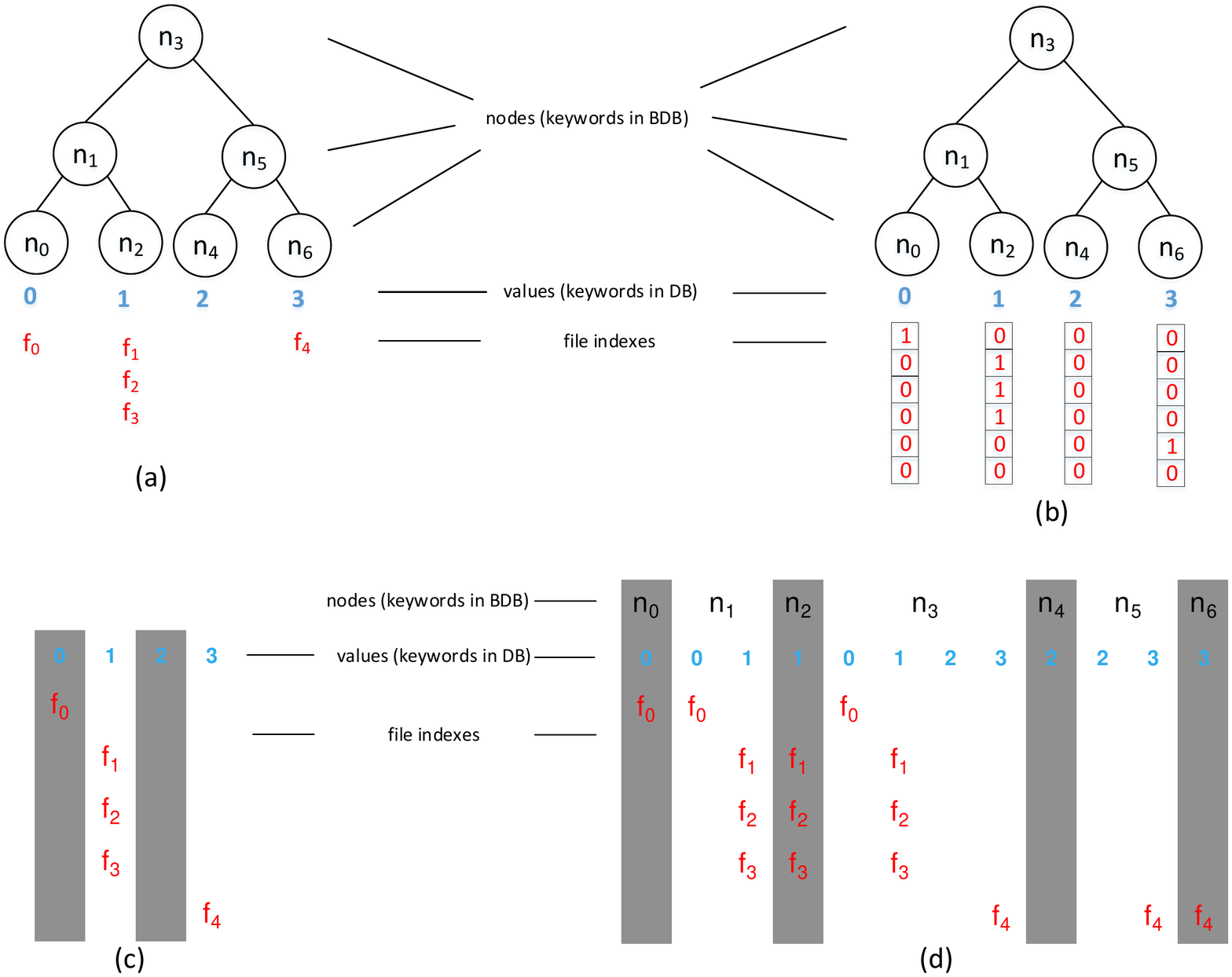}
	\caption{Architecture of Our Binary Tree for Range Queries}
	\label{fig:BT}
\end{figure}

\subsubsection{Binary Tree Assignment and Update.}
As we use the binary tree to support data structure needed in our DSSE, we define the following operations that are necessary to manipulate the DSSE data structure.

	\noindent\textbf{TCon}($m$): For an integer $m$, the operation builds a complete binary tree \texttt{CBT}. \texttt{CBT} has $\lceil log(m)\rceil + 1$ levels, where the root is on the level 0, and the leaves are on the level $\lceil log(m)\rceil$. All leaves are associated with the $m$ consecutive integers from left to right.

	\noindent\textbf{TAssign}(\texttt{CBT}): The operation takes a \texttt{CBT} as an input and outputs an assigned binary tree \texttt{ABT}, where nodes are labelled by appropriate integers. The operation applies \textbf{TAssignSub} recursively. Keywords then are assigned to the node integers.

	\noindent\textbf{TAssignSub}($c$, \texttt{CBT}): For an input pair: a counter $c$ and \texttt{CBT}, the operation outputs an assigned binary tree. It is implemented as a recursive function. It starts from 0 and assigns to nodes incrementally. See Fig. \ref{fig:StpBT} for an example.

\begin{algorithm}[!htb]
	\caption{Our Binary Tree}\label{alg:tree}
	\begin{multicols}{2}
	\underline{\textbf{TCon}($m$)}\\
	\textbf{Input} integer $m$\\
	\textbf{Output} complete binary tree \texttt{CBT}
	\begin{algorithmic}[1]
		\State Construct a \texttt{CBT} with $\lceil log(m)\rceil + 1$ levels.
		\State Set the number of leaves to $m$.
		\State Associate the leaves with $m$ consecutive integers [0,$m$-1] from left to right.
		\State \Return \texttt{CBT}
	\end{algorithmic}

	\underline{\textbf{TAssign}(\texttt{CBT})}\\
    \textbf{Input} complete binary tree \texttt{CBT}\\
    \textbf{Output} assigned binary tree \texttt{ABT}
	\begin{algorithmic}[1]
		\State Counter $c=0$
		\State \textbf{TAssignSub}($c$, \texttt{CBT})
		\State \Return \texttt{ABT}
	\end{algorithmic}

	\underline{\textbf{TAssignSub}($c$, \texttt{CBT})}\\
	\textbf{Input} \texttt{CBT}, counter $c$\\
	\textbf{Output} Assigned binary tree \texttt{ABT}
	\begin{algorithmic}[1]
		\If {\texttt{CBT}.$left\ne\bot$ }
			\State \textbf{TAssignSub}($c$, \texttt{CBT}.$left$)
		\EndIf
		\State Assign \texttt{CBT} with counter $c$.
		\State $c=c+1$
		\If {\texttt{CBT}.$right\ne\bot$}
			\State \textbf{TAssignSub}($c$, \texttt{CBT}.$right$)
		\EndIf
		\State Assign \texttt{CBT} with counter $c$.
		\State $c=c+1$
		\State \Return \texttt{ABT}
	\end{algorithmic}
	
	\underline{\textbf{TGetNodes}($n$, \texttt{ABT})}\\
	\textbf{Input} node $n$, \texttt{ABT}\\
	\textbf{Output} \textbf{NSet}
	\begin{algorithmic}[1]
		\State \textbf{NSet} $\leftarrow$ Empty Set
		\While{$n\ne \perp$}
		\State \textbf{NSet} $\leftarrow$ \textbf{NSet} $\cup$ $n$
        \State $n=n.parent$
		\EndWhile
		\State \Return \textbf{NSet}
	\end{algorithmic}
	
	\underline{\textbf{TUpdate}($add, v$, \texttt{CBT})}\\
	\textbf{Input} op$=add$, value $v$, \texttt{CBT}\\
	\textbf{Output} updated \texttt{CBT}
	\begin{algorithmic}[1]
		\If{\texttt{CBT} $=\bot$}
			\State Create a node.
			\State Associate value $v=0$ to this node.
			\State Set \texttt{CBT} to this node.
		\ElsIf{\texttt{CBT} is \texttt{PBT} or \texttt{CBT} has one node}
			\State Create a new root node $\texttt{root}_n$.
			\State Create a \texttt{VBT} $=$ \texttt{CBT}
			\State \texttt{CBT}$.parent$=\texttt{VBT}.$parent$=$\texttt{root}_n$
			\State \texttt{CBT} $=\texttt{root}_n$
			\State Associate $v$ to the least virtual leaf and set this leaf and its parents as real.
		\Else
			\State Execute line 10.
		\EndIf
		\State \Return \texttt{CBT}
	\end{algorithmic}
	\end{multicols}
\end{algorithm}

	\noindent\textbf{TGetNodes}($n$, \texttt{ABT}): For an input pair: a node $n$ and a tree \texttt{ABT}, the operation generates a collection of nodes in a path from the node $n$ to the root node. This operation is needed for our update algorithm if a client wants to add a file to a leaf (a value in the range). The file is added to the leaf and its parent nodes.

	\noindent\textbf{TUpdate}(add, $v$, \texttt{CBT}): The operation takes a value $v$ and a complete binary tree \texttt{CBT} and updates \texttt{CBT} so the tree contains the value $v$. For simplicity, we consider the current complete binary tree contains values in the range $[0,v-1]$\footnote{Note that, we can use \textbf{TUpdate} many times if we need to update more values.}. Depending on the value of $v$, the operation is executed according to the following cases:
	
	\begin{itemize}
		\item $v=0$: It means that the current complete binary tree is null, we simply create a node and associate value $v=0$ with the node. The operation returns the node as \texttt{CBT}. 
		
		\item $v>0$: If the current complete binary tree is a perfect binary tree \texttt{PBT} or it consists of a single node only, we need to create a virtual binary tree \texttt{VBT}, which is a copy of the current binary tree. Next, we merge the virtual perfect binary tree with the original one getting a large perfect binary tree. Finally, we need to associate the value $v$ with the least virtual leaf (the leftmost virtual leaf without a value) of the virtual binary tree and set this leaf and its parents as real. For example, in Fig. \ref{fig:StpBT}.(a), $v=4$, the nodes with solid line are real and the nodes with dot line are virtual which can be added later. Otherwise, we directly associate the value $v$ to the least virtual leaf and set this leaf and its parents as real \footnote{Only if its parents were virtual, then we need to convert them to real.}. In Fig. \ref{fig:StpBT}.(b), $v=5$.
	
	\end{itemize}

\begin{figure}[!htbp]
	\centering
	\includegraphics[width=\textwidth,trim=0 150 0 150]{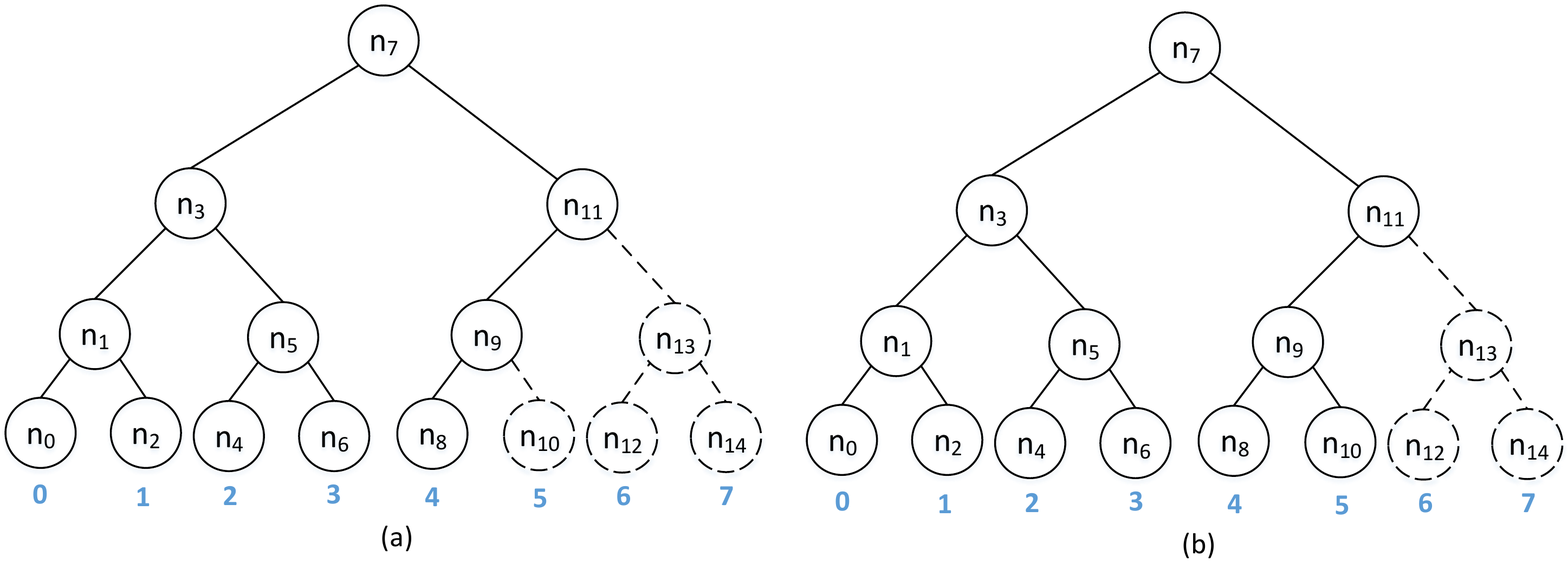}
	\caption{Example of Update Operation}
	\label{fig:StpBT}
\end{figure}

\noindent Note that, in our range queries, we need to parse a normal database \texttt{DB} to its binary form \texttt{BDB}. First, we need to map keywords of \texttt{DB} to integers in the range $[0, |W|-1]$, where $|W|$ is the total number of keywords in \texttt{DB}. Next, we construct a binary tree as described above. The keywords are assigned to the nodes of the binary tree and are associated with the documents of their descendants. For example, In Fig. \ref{fig:BT}.a, the keywords are $\{n_0, n_1, \cdots, n_6\}$ and $\texttt{BDB}(n_0)=\{f_0\}$, $\texttt{BDB}(n_1)=\{f_0,f_1,f_2,f_3\}$.

\subsection{DSSE Range Queries - Construction A}
In this section, we apply our new binary tree to the Bost \cite{Bos16} scheme to support range queries. For performing a ranger query, the client in our scheme first determine a collection of keywords to cover the requested range. Then, she generates the search token corresponding to each node (in the cover) and sends them to the sever, which can be done in a similar way as \cite{Bos16}. Now we are ready to present the first DSSE scheme that supports range queries and is forward-private. The scheme is described in Algorithm \ref{alg:Bost}, where $F$ is a cryptographically strong pseudorandom function (PRF), $H_1$ and $H_2$ are keyed hash functions and $\Pi$ is a trapdoor permutation.  

	\noindent \textbf{Setup}($1^\lambda$): For a security parameter $1^\lambda$ , the algorithm outputs ($\texttt{TPK}, \texttt{TSK},$ $K, \textbf{T}, \textbf{N}, m$), where $\texttt{TPK}$ and $\texttt{TSK}$ are the public key and secret keys of the trapdoor permutation, respectively, $K$ is the secret key of function $F$, \textbf{T}, \textbf{N} are maps and $m$ is the maximum number of the values in our range queries. The map \textbf{N} is used to store the pair keyword/($ST_c$, $c$) (current search token and the counter $c$, please see Algorithm \ref{alg:Bost} for more details.) and is kept by the client. The map \textbf{T} is the encrypted database \texttt{EDB} that used to store the encrypted indices which is kept by the server.

	\noindent \textbf{Search}($[a,b],\sigma,$ $m$, \texttt{EDB}): The protocol is executed between a client and a server. The client asks for documents, whose keywords are in the range $[a,b]$, where $0\le a \le b<m$. The current state of \texttt{EDB} is $\sigma$ and the integer $m$ describes the maximum number of values. Note that knowing $m$, the client can easily construct the complete binary tree. The server returns a collection of file indices of requested documents.

\begin{algorithm}[!htb]
    \caption{Construction A}\label{alg:Bost}
    \begin{multicols}{2}
    \underline{\textbf{Setup}($1^{\lambda}$)}\\
    \textbf{Input} security parameter $1^{\lambda}$\\
    \textbf{Output} $(\texttt{TPK}, \texttt{TSK}, K, \textbf{T}, \textbf{N}, m)$
    %\textit{Client:}
    \begin{algorithmic}[1]
        \State $K\leftarrow\{0,1\}^{\lambda}$
        \State $(\texttt{TSK}, \texttt{TPK})\leftarrow \texttt{TKeyGen}(1^{\lambda})$
        \State \textbf{T}, \textbf{N} $\leftarrow$ empty map
        \State $m=0$
        \State \Return $(\texttt{TPK}, \texttt{TSK}, K, \textbf{T}, \textbf{N}, m)$
    \end{algorithmic}

    \underline{\textbf{Search}($[a,b],\sigma,$ $m$, \texttt{EDB})}\\
    \textit{Client:}\\
    \textbf{Input} $[a,b],\sigma,$ $m$\\
    \textbf{Output} $(K_n,ST_c,c)$
    \begin{algorithmic}[1]
    	\State \texttt{CBT} $\leftarrow$ \textbf{TCon}($m$)
    	\State \texttt{ABT} $\leftarrow$ \textbf{TAssign}(\texttt{CBT})
    	\State \textbf{RSet} $\leftarrow$ Find the minimum nodes to cover $[a,b]$ in \texttt{ABT}
    	\For{$n\in \textbf{RSet}$}
        	\State $K_{n}\leftarrow F_{K}(n)$
        	\State $(ST_c,c)\leftarrow\textbf{N}[n]$
        	\If{$(ST_c,c)\ne\perp$}
        		\State Send $(K_n,ST_c,c)$ to the server.
        	\EndIf
        \EndFor
        \algstore{break}
    \end{algorithmic}
	
	\textit{Server:}\\
	\textbf{Input} $(K_n,ST_c,c),$ \texttt{EDB}\\
	\textbf{Output} ($ind$)
	\begin{algorithmic}[1]
		\algrestore{break}
		\State Upon receiving $(K_n,ST_c,c)$
		\For{$i=c$ to 0}
			\State $UT_i\leftarrow H_1(K_n,ST_i)$
			\State $e\leftarrow \textbf{T}[UT_i]$
			\State $ind \leftarrow e\oplus H_2(K_n,ST_i)$
			\State Output the $ind$
			\State $ST_{i-1}\leftarrow\Pi(\texttt{TPK},ST_i)$
		\EndFor
	\end{algorithmic}

    \underline{\textbf{Update}($add, v, ind,\sigma,$ $m$, \texttt{EDB})}\\
	\textit{Client:}\\
	\textbf{Input} $add, v, ind,\sigma,$ $m$\\
	\textbf{Output} $(UT_{c+1},e)$
	\begin{algorithmic}[1]
		\State \texttt{CBT} $\leftarrow$ \textbf{TCon}($m$)
        \If{$v=m$}
            \State \texttt{CBT}$\leftarrow$\textbf{TUpdate}($add, v,$ \texttt{CBT})
        	\State $m\leftarrow m+1$
        	\If{\texttt{CBT} added a new root}
                \State $(ST_c,c)\leftarrow \textbf{N}[\texttt{root}_o]$
                \State $\textbf{N}[\texttt{root}_n]\leftarrow(ST_c,c)$
        	\EndIf

        \State Get the leaf $n_v$ of value $v$.
        \State \texttt{ABT} $\leftarrow$ \textbf{TAssign}(\texttt{CBT})
		\State $\textbf{NSet}\leftarrow\textbf{TGetNodes}(n_v$, \texttt{ABT})
		\For {every node $n\in \textbf{NSet}$}
            \State $K_n\leftarrow F_{K}(n)$
		    \State $(ST_c,c)\leftarrow\textbf{N}[n]$
		    \If{$(ST_c,c)=\perp$}
			     \State $ST_0\leftarrow \mathcal{M},c\leftarrow -1$
		    \Else
			     \State $ST_{c+1}\leftarrow\Pi^{-1}(\texttt{TSK},ST_c)$
		    \EndIf
	
		\State $\textbf{N}[n]\leftarrow (ST_{c+1}, c+1)$
		\State $UT_{c+1}\leftarrow H_1(K_n,ST_{c+1})$
		\State $e\leftarrow ind\oplus H_2(K_n,ST_{c+1})$
		\State Send $(UT_{c+1},e)$ to the Server.
        \EndFor
        \ElsIf{$v<m$}
            \State Execute line 9-24.
        \EndIf
    \algstore{break}
	\end{algorithmic}

	\textit{Server:}\\
	\textbf{Input} $(UT_{c+1},e),$ \texttt{EDB}\\
	\textbf{Output} \texttt{EDB}
	\begin{algorithmic}[1]
		\algrestore{break}
        \State Upon receiving $(UT_{c+1},e)$
		\State Set $\textbf{T}[UT_{c+1}]\leftarrow e$
	\end{algorithmic}
	\end{multicols}
\end{algorithm}

	\noindent \textbf{Update}($add, v, ind, \sigma,$ $m$, \texttt{EDB}): The protocol is performed jointly by a client and server. The client wishes to add an integer $v$ together with a file index $ind$ to \texttt{EDB}. The state of \texttt{EDB} is $\sigma$, the number of values $m$. There are following three cases:

\begin{itemize}
	\item $v<m$: The client simply adds $ind$ to the leaf, which contains value $v$ and its parents (See line 9-24 in Algorithm \ref{alg:Bost}). This is a basic update, which is similar to the one from \cite{Bos16}.

	\item $v=m$: The client first updates the complete binary tree to which she adds the value $v$. If a new root is added to the new complete binary tree, then the server needs to add all file indices of the old complete binary tree to the new one. Finally, the server needs to add $ind$ to the leaf, which contains value $v$ and its parents.

	\item $v>m$: The client uses \textbf{Update} as many times as needed. For simplicity, we only present the simple case $v=m$, i.e., the newly added value $v$ equals the maximum number of values of the current range $[0, m-1]$, in the description of Algorithm \ref{alg:Bost}. 
\end{itemize}

\noindent The DSSE supports range queries at the cost of large client storage, since the number of search tokens is linear in the number of all nodes of the current tree instead of only leaves. In \cite{Bos16}, the number of entries at the client is $|W|$, while it would be roughly $2|W|$ in this construction. Moreover, the communication cost is heavy since the server needs to return all file indices to the client for every search.  To overcome the weakness, we give a new construction with lower client storage and communication cost in the following section.

\subsection{DSSE Range Queries - Construction B} \label{subsec:b}
In this section, we give the second construction by leveraging the Paillier cryptosystem \cite{Pai99} and bit string representation, which significantly reduce the client storage and communication cost compared with the first one at the cost of losing forward privacy. With the the homomorphic addition property of the Paillier cryptosystem, we can add and delete the file indices by parsing them into binary strings, as illustrated in Section \ref{subsubsec:bs}. Next we briefly describe our second DSSE, which can not only support range queries but also achieve backward privacy. The scheme is described in Algorithm \ref{alg:Paillier}. 

	\noindent \textbf{Setup}($1^\lambda$): For a security parameter $1^\lambda$ , the algorithm returns ($\texttt{PK}, \texttt{SK}, K,$ $\textbf{T}, m$), where $\texttt{PK}$ and $\texttt{SK}$ are the public and secret keys of the Paillier cryptosystem, respectively, $K$ is the secret key of a PRF $F$, $m$ is the maximum number of values which can be used to reconstruct the binary tree and the encrypted database \texttt{EDB} is stored in a map \textbf{T} which is kept by the server.

	\noindent \textbf{Search}($[a,b],\sigma,$ $m$, \texttt{EDB}): The protocol is executed between a client and a server. The client queries for documents, whose keywords are in the range $[a,b]$, where $0\le a \le b<m$. $\sigma$ is the state of \texttt{EDB}, and integer $m$ specifies the maximum values for our range queries. The server returns encrypted file indices $e$ to the client, who can decrypt $e$ by using the secret key $\texttt{SK}$ of Pailler Cryptosystem and obtain the file indices of requested documents.

	\noindent \textbf{Update}($op, v, ind, \sigma,$ $m$, \texttt{EDB}): The protocol runs between a client and a server. A requested update is named by the parameter $op$. The integer $v$ and the file index $ind$ specifies the tree nodes that need to be updated. The current state $\sigma$, the integer $m$ and the server with input \texttt{EDB}. If $op=add$, the client generates a bit string as prescribed in Section \ref{subsubsec:bs}. In case when $op=delete$, the client creates the complement bit string as given in Section \ref{subsubsec:bs}. The bit string $bs$ is encrypted using the Paillier cryptosystem. The encrypted string is denoted by $e$. There are following three cases:
	
	\begin{algorithm}[!htb]
		\caption{Construction B}\label{alg:Paillier}
		\begin{multicols}{2}
			\underline{\textbf{Setup}($1^{\lambda})$}\\
			\textbf{Input} security parameter $1^{\lambda}$\\
			\textbf{Output} $(\texttt{PK}, \texttt{SK}, K, \textbf{T}, m)$
			\begin{algorithmic}[1]
				\State $K\leftarrow\{0,1\}^{\lambda}$
				\State $(\texttt{SK}, \texttt{PK})\leftarrow \texttt{KeyGen}(1^{\lambda})$
				\State \textbf{T} $\leftarrow$ empty map
				\State $m=0$
				\State \Return $(\texttt{PK}, \texttt{SK}, K, \textbf{T}, m)$
			\end{algorithmic}
			
			\underline{\textbf{Search}($[a,b],\sigma,$ $m$, \texttt{EDB})}\\
			\textit{Client:}\\
			\textbf{Input} $[a,b],\sigma,m$\\
			\textbf{Output} $(UT_n)$
			\begin{algorithmic}[1]
				\State \texttt{CBT} $\leftarrow$ \textbf{TCon}($m$)
				\State \texttt{ABT} $\leftarrow$ \textbf{TAssign}(\texttt{CBT})
				\State \textbf{RSet} $\leftarrow$ Find the minimum nodes to cover $[a,b]$ in \texttt{ABT}
				\For{$n\in \textbf{RSet}$}
				\State $UT_n\leftarrow F_{K}(n)$
				\State Send $UT_n$ to the server.
				\EndFor
				\algstore{break}
			\end{algorithmic}
			
			\textit{Server:}\\
			\textbf{Input} $(UT_n),$ \texttt{EDB}\\
			\textbf{Output} $(e)$
			\begin{algorithmic}[1]
				\algrestore{break}
				\State Upon receiving $UT_n$
				\State $e\leftarrow \textbf{T}[UT_n]$
				\State Send $e$ to the Client.
				%\algstore{break}
			\end{algorithmic}
			
			\underline{\textbf{Update}($op,v,ind,\sigma,$ $m$, \texttt{EDB})}\\
			\textit{Client:}\\
			\textbf{Input} $op,v,ind,\sigma,m$\\
			\textbf{Output} $(UT_n,e)$
			\begin{algorithmic}[1]
				\State \texttt{CBT} $\leftarrow$ \textbf{TCon}($m$)
				\If{$v=m$}
				\State \texttt{CBT} $\leftarrow$ \textbf{TUpdate}($add,v,$ \texttt{CBT})
				\State $m\leftarrow m+1$
				\If{\texttt{CBT} added a new root}
				\State $UT_{\texttt{root}_o}\leftarrow F_K(\texttt{root}_o)$
				\State $UT_{\texttt{root}_n}\leftarrow F_K(\texttt{root}_n)$
				\State $e\leftarrow\textbf{T}[UT_{\texttt{root}_o}]$
				\State $\textbf{T}[UT_{\texttt{root}_n}]\leftarrow e$
				\EndIf
				\State Get the leaf $n_v$ of value $v$.
				\State \texttt{ABT} $\leftarrow$ \textbf{TAssign}(\texttt{CBT})
				\State $\textbf{NSet}\leftarrow\textbf{TGetNodes}(n_v, \texttt{ABT})$
				\If{$op=add$}
				\State Generate the bit string $bs$ as state in Bit String Representation of Section \ref{subsubsec:bs}.
				\ElsIf{$op=del$}
				\State Generate the complement bit string $bs$ as state in Bit String Representation of Section \ref{subsubsec:bs}.
				\EndIf
				\For {every node $n\in \textbf{NSet}$}
				\State $UT_n\leftarrow F_K(n)$
				\State $e\leftarrow \texttt{Enc}(\texttt{PK}, bs)$
				\State Send $(UT_n,e)$ to the server.
				\EndFor
				\ElsIf{$v<m$}
				\State Execute line 11-23.
				\EndIf
			\end{algorithmic}
			
			\textit{Server:}\\
			\textbf{Input} $(UT_n,e),$ \texttt{EDB}\\
			\textbf{Output} \texttt{EDB}
			\begin{algorithmic}[1]
				\State Upon receiving $(UT_n,e)$
				\State $e'\leftarrow \textbf{T}[UT_n]$
				\If{$e'\ne \perp$}
				\State $e\leftarrow e\cdot e'$
				\EndIf
				\State $\textbf{T}[UT_n]\leftarrow e$
			\end{algorithmic}
			
		\end{multicols}
	\end{algorithm}

\begin{itemize}
	\item $v<m$: The client sends the encrypted bit string $e$ with the leaf $n_v$ containing value $v$ and its parents to server. Next the server adds $e$ with the existing encrypted bit strings corresponding to the nodes specified by the client. See line 11-23 in Algorithm \ref{alg:Paillier} which is similar to the update in Algorithm \ref{alg:Bost}.
	
	\item $v=m$: The client first updates the complete binary tree to which she adds the value $v$. If a new root is added to the new complete binary tree, then the client retrieves the encrypted bit string of the root (before update). Next the client adds it to the new root by sending it with the new root to the server. Finally, the client adds $e$ to the leaf that contains value $v$ and its parents as in $v<m$ case.
	
	\item $v>m$: The client uses \textbf{Update} as many times as needed. For simplicity, we only consider $v=m$, where $m$ is the number of values in the maximum range.
\end{itemize}

\noindent In this construction, it achieves backward privacy. Moreover, the communication overhead between the client and the server is significantly reduced due to the fact that for each query, the server returns a single ciphertext to the client at the cost of supporting small number of documents. Since, in Paillier cryptosystem, the length of the message is usually small and fixed (e.g. 1024 bits).

This construction can be applied to applications, where the number of documents is small and simultaneously the number of keywords can be large. The reason for this is the fact that for a given keyword, the number of documents which contain it is small. Consider a temperature forecast system that uses a database, which stores records from different sensors (IoT) located in different cities across Australia. In the application, the cities (sensors) can be considered as documents and temperature measurements can be considered as the keywords. For example, Sydney and Melbourne have the temperature of 18$^{\circ}$C. Adelaide and Wollongong have got 17$^{\circ}$C and 15$^{\circ}$C, respectively. If we query for cities, whose temperature measurements are in the range from 17 to 18$^{\circ}$C, then the outcome includes Adelaide, Sydney and Melbourne. Here, the number of cities (documents) is not large. The number of different temperature measurements (keywords) can be large depending on requested precision.

\section{Security Analysis}\label{sec:security}
Similar to \cite{FJKNRS15}, for a range query $q=[a,b]$, let $\{n_{c_1},...,n_{c_t}\}$ be the tree cover of interval $[a,b]$. We consider $n_{c_i}$ as a keyword and parse a range query into several keywords. Before define the leakage functions, we define a search query $q=(t,[a,b])=\{(t,n_{c_1}),\cdots,(t,n_{c_t}) \}$. For an update query, if we want to update a file $ind$ with value $v$, we may need to update the corresponding leaf node and its parents in the tree denoted as $\{n_{u_1},\cdots,n_{u_t}\}$. We define an update query $u=(t,op,(v,ind))=\{(t,op,(n_{u_1},ind)),\cdots,(t,op,(n_{u_t},ind))\}$. For a list of search query $Q=\{(t,n) : (t,n)\in \{q\}\}$ and a list of update query $Q'=\{(t,op,(n,ind)):(t,op,(n,ind))\in \{u\}\}$ Then, following \cite{Bos16}, the leakage to the server is summarized as follows:
\begin{itemize}
	\item Search pattern $\texttt{sp}(n)=\{t:(t,n)\in Q\}$, it leaks the timestamp $t$ that the same search query on $n$.
	\item History $\texttt{Hist}(n)=\{(t,op,ind) : (t,op,(n,ind))\in Q'\}$, the history of keyword $n$. It includes all the updates made to $\texttt{DB}(n)$ and when the update happened.
	\item contain pattern $\texttt{cp}(n)=\{t': \texttt{DB}(n)\subseteq \texttt{DB}(n')$ and $t'<t, (t',n'),(t,n)\in Q \}$, it leaks the time $t'$ of previous search query on keyword $n'$, where $\texttt{DB}(n)\subseteq \texttt{DB}(n')$. Note that, $\texttt{cp}(n)$ is an inherited leakage for range queries when the file indices are revealed to the server. If a query $n$ is a subrange of query $n'$, then the file index set for $n$ will also be a subset of the file index set for $n'$.
	%\item time $\texttt{HistUp}(n)=\{t_i,ut_i\}_{i=1}^{u_n}$, where $u_n$ is the total number of updates made to keyword $n$ and $ut_i$ is the update token for each update.
\end{itemize} 
%Note that, contain pattern $\texttt{cp}(w)$ is an inherited leakage for range queries when the file indices are revealed to the server. If a query $w'$ is a subrange of query $w$, then the file index set for $w'$ will also be a subset of the file index set for $w$.

\subsection{Forward Privacy}
Following \cite{Bos16}, forward privacy means that an update does not leak any information about keywords of updated documents matching a query we previously issued.  A formal definition is given below:

\begin{definition}(\cite{Bos16})\label{def:fs}
	A $\mathcal{L}$-adaptively-secure DSSE scheme $\Gamma$ is forward-private if the update leakage function $\mathcal{L}^{Updt}$ can be written as $$\mathcal{L}^{Updt}(op,in)=\mathcal{L}'(op,{(ind_i,\mu_i)})$$ where ${(ind_i,\mu_i)}$ is the set of modified documents paired with number $\mu_i$ of modified keywords for the updated document $ind_i$.
\end{definition}

%Note that, our second construction can only achieve a weak forward security. Because it leaks the update corresponding to a keyword but not the keyword/document pairs. A formal definition is given below:

\begin{comment}
\begin{definition}\label{def:fs2}
	A $\mathcal{L}$-adaptively-secure DSSE scheme $\Gamma$ is weak forward-private if the update leakage function $\mathcal{L}^{Updt}$ can be written as $$\mathcal{L}^{Updt}(op,in)=\mathcal{L}'(op, \texttt{Time}(w))$$ where $\texttt{Time}(w)$ is the number of updates made to keyword $w$ and when the update happened.
\end{definition}
\end{comment}

\subsection{Construction A}
Since the first DSSE construction is based on \cite{Bos16}, it inherits security of the original design. Adaptive security of the construction A can be proven in the Random Oracle Model and is a modification of the security proof of \cite{Bos16}. We refer readers to Appendix for the full proof.

\begin{theorem}\label{th:af}
(Adaptive forward privacy of A). Let $\mathcal{L}_{\Gamma_A}=(\mathcal{L}_{\Gamma_A}^{Srch}$, $\mathcal{L}_{\Gamma_A}^{Updt})$, where $\mathcal{L}_{\Gamma_A}^{Srch}(n)$ $=(\texttt{sp}(n),\texttt{Hist}(n),\texttt{cp}(n))$, $\mathcal{L}_{\Gamma_A}^{Updt}(add,n,ind)=\perp$. The construction A is $\mathcal{L}_{\Gamma_A}$-adaptively forward-private.
\end{theorem}
	
Compared with \cite{Bos16}, this construction additionally leaks the contain pattern $\texttt{cp}$ as described in Section \ref{subsec:security}. Other leakages are exactly the same as \cite{Bos16}. Since the server executes one keyword search and update one keyword/file-index pair at a time. Note that the server does not know the secret key of the trapdoor permutation, so it cannot learn anything about the pair even if the keyword has been searched by the client previously.

\subsection{Backward Privacy}
%Before giving the definition of backward privacy, we define a new leakage function $\texttt{HistUp}=\{t: (t,op,(n,bs))\in Q'\}$ which leaks the timestamp $t$ that when the update happens corresponding to keyword $n$. In Construction B, we use a bit string $bs$ to represent file identifiers.
Backward privacy means that a search query on keyword $n$ does not leak the file indices that previously added and later deleted. More formally, we modify the Type I definition of \cite{BMO17}. It leaks keyword $n$ has been updated\footnote{Instead of leaking the keyword $n$ in the plaintext form, it may be leaked in the masked form.}, the total number of updates on $n$. More formally,
\begin{definition}\label{def:bs}
	A $\mathcal{L}$-adaptively-secure DSSE scheme $\Gamma$ is backward-private if the the search and update leakage functions $\mathcal{L}^{Srch}$, $\mathcal{L}^{Updt}$ can be written as  $\mathcal{L}^{Updt}(n)=\mathcal{L}'(n)$, $\mathcal{L}^{Srch}=\mathcal{L}''(\texttt{sp}(n))$.
\end{definition}

\subsection{Construction B}
The adaptive security of second DSSE construction relies on the semantic security of Paillier cryptosystem. All file indices are encrypted using the public key of Paillier cryptosystem. Without the secret key, the server cannot learn anything from the ciphertext. Note that, during update, this construction leaks which keyword has been updated. We refer readers to Appendix for the full proof.

\begin{comment}
\begin{theorem}\label{th:bf}
(Adaptive weak forward privacy of B). Let $\mathcal{L}_{\Gamma_B}=(\mathcal{L}_{\Gamma_B}^{Srch}$, $\mathcal{L}_{\Gamma_B}^{Updt})$, where $\mathcal{L}_{\Gamma_B}^{Srch}(n)$ $=(\texttt{sp}(n))$, $\mathcal{L}_{\Gamma_B}^{Updt}(op,n,ind)=(\texttt{Time}(n))$. Construction B is $\mathcal{L}_{\Gamma_B}$-adaptively forward-private.
\end{theorem}

\begin{proof}
(Sketch) In construction B, for the update, we only leak the number of updates corresponding to the queried keywords \textbf{n} rather than hide the keyword/document pairs as in \cite{Bos16}. Since all cryptographic operations are performed at the client side where no keys are revealed to the server, the server can learn nothing from the update, given that the Paillier cryptosystem scheme is IND-CPA secure. We can simulate the $\verb"DSSEREAL"$ as in Algorithm \ref{alg:Paillier} and simulate the $\verb"DSSEIDEAL"$ by encrypting all 0's strings for the \texttt{EDB}. The adversary $\mathcal{A}$ can not distinguish the real ciphertext from the ciphertext of 0's. Then, $\mathcal{A}$ cannot distinguish $\verb"DSSEREAL"$ from $\verb"DSSEIDEAL"$. Hence, our Construction B achieves weak forward privacy.\hfill{$\Box$}
\end{proof}
\end{comment}

\begin{theorem}\label{th:bB}
(Adaptive backward privacy of B).  Let $\mathcal{L}_{\Gamma_B}=(\mathcal{L}_{\Gamma_B}^{Srch}$, $\mathcal{L}_{\Gamma_B}^{Updt})$, where $\mathcal{L}_{\Gamma_B}^{Srch}(n)=(\texttt{sp}(n))$, $\mathcal{L}_{\Gamma_B}^{Updt}(op,n,ind)=(n)$. Construction B is $\mathcal{L}_{\Gamma_B}$-adaptively backward-private.
\end{theorem}

During the update, the construction B leaks which keyword has been updated. However, it does not leak the type of update (either add or del) on encrypted file indices, because both addition and deletion are achieved by homomorphic addition. Moreover, it does not leak contain pattern $\texttt{cp}$ and the file indices that previously added and later deleted since the file indices have been encrypted, and the server can learn nothing without the secret key.

\section{Conclusion}\label{sec:conclusion}
In this paper, we give two secure DSSE schemes that support range queries. The first DSSE construction applies our binary tree to Bost \cite{Bos16}'s framework which achieves forward privacy. However, it incurs a large storage overhead in the client and a large communication cost between the client and the server. To achieve backward privacy, we propose the second DSSE construction with range queries by applying Paillier cryptosystem and bit string representation. In this construction, we use the fixed update token to reduce the client and the server storage at the cost of losing forward privacy. In addition, it can not support large number of documents. Although the second DSSE construction cannot support large number of documents, it can still be very useful in certain applications. In the future, we would like to construct more scalable DSSE schemes with more expressive queries.

\section*{Acknowledgment}
The authors thank the anonymous reviewers for the valuable comments. This work was supported by the Natural Science Foundation of Zhejiang Province [grant number LZ18F020003], the National Natural Science Foundation of China [grant number 61472364] and the Australian Research Council (ARC) Grant DP180102199. Josef Pieprzyk has been supported by National Science Centre, Poland, project registration number UMO-2014/15/B/ST6/05130.

\bibliographystyle{splncs04}
\bibliography{drq}

\section*{Appendix}

\textbf{Theorem \ref{th:af}.} \emph{Let $\Pi$ be the one-way trapdoor permutation, $F$ be a PRF, and $H_1$, $H_2$ be the hash functions and manipulated as random oracles outputting $\lambda$ bits. Define $\mathcal{L}_{\Gamma_A}=(\mathcal{L}_{\Gamma_A}^{Srch}$, $\mathcal{L}_{\Gamma_A}^{Updt})$, where $\mathcal{L}_{\Gamma_A}^{Srch}(n)=(\texttt{sp}(n),\texttt{Hist}(n),\texttt{cp}(n))$, $\mathcal{L}_{\Gamma_A}^{Updt}(add,n,ind)=\perp$. Construction A is $\mathcal{L}_{\Gamma_A}$-adaptively forward-private.}

\begin{proof} As mentioned before, we parse a range interval into several keywords. Following \cite{Bos16}, we will set a serial of games from $\verb"DSSEREAL"_{\mathcal{A}}^{\Gamma_A}(1^{\lambda})$ to $\verb"DSSEIDEAL"_{\mathcal{A},\mathcal{S}_1}^{\Gamma_A}(1^{\lambda})$.
	
	\textbf{Game} $G_{1,0}$: $G_{1,0}$ is exactly same as the real world game $\verb"DSSEREAL"_{\mathcal{A}}^{\Gamma_A}(1^{\lambda})$.$$\Pr[\verb"DSSEREAL"_{\mathcal{A}}^{\Gamma_A}(1^{\lambda})=1]=\Pr[G_{1,0}=1].$$
	
	\textbf{Game} $G_{1,1}$: Instead of calling $F$ when generating $k_n$, $G_{1,1}$ picks a new random key when it inputs a new keyword $n$, and stores it in a table $\texttt{Key}$ so it can be reused next time. If an adversary $\mathcal{A}$ is able to distinguish between $G_{1,0}$ and $G_{1,1}$, we can then build a reduction to distinguish between $F$ and a truly random function. More formally, there exists an efficient adversary $\mathcal{B}_1$ such that $$\Pr[G_{1,0}=1]-\Pr[G_{1,1}=1]\le \verb"Adv"_{F,\mathcal{B}_1}^{\texttt{prf}}(\lambda).$$
	
	\begin{algorithm}[!htb]
		\caption{\textbf{Game} $G_{1,2}$ and single box for $G_{1,2}'$} \label{alg:aG1}
		\begin{multicols}{2}
			\underline{\textbf{Setup}($1^{\lambda}$)}
			%\textit{Client:}
			\begin{algorithmic}[1]
				\State $(\texttt{TSK}, \texttt{TPK})\leftarrow \texttt{TKeyGen}(1^{\lambda})$
				\State \textbf{T}, \textbf{N} $\leftarrow$ empty map
				\State $m=0$
				\State $bad\leftarrow false$
				\State \Return $(\texttt{TPK}, \texttt{TSK}, K, \textbf{T}, \textbf{N}, m)$
			\end{algorithmic}
			
			\underline{\textbf{Search}($[a,b],\sigma;$ $m$, EDB)}\\
			\textit{Client:}
			\begin{algorithmic}[1]
				\State CBT $\leftarrow$ \textbf{TCon}($m$)
				\State ABT $\leftarrow$ \textbf{TAssign}(CBT)
				\State \textbf{RSet} $\leftarrow$ \textbf{TGetCover}($[a,b],$ ABT)
				\For{$n\in \textbf{RSet}$}
				\State $K_{n}\leftarrow \texttt{Key}(n)$
				\State $(ST_0,\cdots,ST_c,c)\leftarrow\textbf{N}[n]$
				\If{$(ST_c,c)\ne\perp$}
				\For{$i=0$ to $c$}
				\State $H_1(K_n,ST_i)\leftarrow \texttt{UT}[n,i]$
				\EndFor
				\State Send $(K_n,ST_c,c)$ to the server.
				\EndIf
				\EndFor
				\algstore{break}
			\end{algorithmic}
			
			\textit{Server:}
			\begin{algorithmic}[1]
				\algrestore{break}
				\State Upon receiving $(K_n,ST_c,c)$
				\For{$i=c$ to 0}
				\State $UT_i\leftarrow H_1(K_n,ST_i)$
				\State $e\leftarrow \textbf{T}[UT_i]$
				\State $ind \leftarrow e\oplus H_2(K_n,ST_i)$
				\State Output the $ind$
				\State $ST_{i-1}\leftarrow\Pi(\texttt{TPK},ST_i)$
				\EndFor
			\end{algorithmic}
			
			\underline{\textbf{Update}($add, v, ind,\sigma;$ $m$, EDB)}\\
			\textit{Client:}
			\begin{algorithmic}[1]
				\State CBT $\leftarrow$ \textbf{TCon}($m$)
				\If{$v=m$}
				\State CBT$\leftarrow$\textbf{TUpdate}($add, v,$ CBT)
				\State $m\leftarrow m+1$
				\If{CBT added a new root}
				\State $(ST_c,c)\leftarrow \textbf{N}[root_{old}]$
				\State $\textbf{N}[root_{new}]\leftarrow(ST_c,c)$
				\EndIf
				
				\State Get the leaf $n_v$ of value $v$.
				\State ABT $\leftarrow$ \textbf{TAssign}(CBT)
				\State $\textbf{NSet}\leftarrow\textbf{TGetNodes}(n_v$, ABT)
				\For {every node $n\in \textbf{NSet}$}
				\State $K_n\leftarrow \texttt{Key}(n)$
				\State $(ST_c,c)\leftarrow\textbf{N}[n]$
				\If{$(ST_c,c)=\perp$}
				\State $ST_0\leftarrow \mathcal{M},c\leftarrow -1$
				\Else
				\State $ST_{c+1}\leftarrow\Pi^{-1}(\texttt{TSK},ST_c)$
				\EndIf
				
				\State $\textbf{N}[n]\leftarrow (ST_0,\cdots, ST_{c+1}, c+1)$
				\State $UT_{c+1}\leftarrow \{0,1\}^x$
				\If{$H_1(K_n,ST_{c+1}\ne \perp)$}
					\State \fbox{$bad\leftarrow true, UT_{c+1}\leftarrow H_1(K_n,ST_{c+1})$}
				\EndIf
				\State $\texttt{UT}[n,c+1]\leftarrow UT_{c+1}$
				%\State \fbox{$UT_{c+1}\leftarrow H_1(K_n,ST_{c+1})$}
				\State $e\leftarrow ind\oplus H_2(K_n,ST_{c+1})$
				\State Send $(UT_{c+1},e)$ to the Server.
				\EndFor
				\ElsIf{$v<m$}
				\State Execute line 9-24.
				\Else
				\State We can use \textbf{Update} many times.
				\EndIf
				\algstore{break}
			\end{algorithmic}
			
			\textit{Server:}
			\begin{algorithmic}[1]
				\algrestore{break}
				\State Upon receiving $(UT_{c+1},e)$
				\State Set $\textbf{T}[UT_{c+1}]\leftarrow e$
			\end{algorithmic}
			
			%\begin{mdframed}
			%\end{mdframed}
			\underline{${H_1}(k, v)$}\\
			\begin{algorithmic}[1]
				\State  $v'\leftarrow H_1(k,v)$
				\If{$v'=\perp$}
				\State $v'\leftarrow\{0,1\}^x$
				\If{$\exists n,c$ s.t $v=ST_c\in\textbf{N}[n]$}
				\State \fbox{$bad\leftarrow true, v'\leftarrow UT[n,c]$}
				\EndIf
				\State $H_1(k,v)\leftarrow v'$
				\EndIf
				\State \Return $v'$
			\end{algorithmic}
			
		\end{multicols}
	\end{algorithm}
	
	\textbf{Game} $G_{1,2}$: In $G_{1,2}$, we pick random strings in replace of calling hash function $H_1$, where $H_1$ is modeled as a random oracle. For every search, the output of $H_1$ is programmed where $H_1(K_n, ST_c(n))=\texttt{UT}[n,c]$.
	
	Algorithm \ref{alg:aG1} describes this game and introduce an intermediate game $G_{1,2}'$. $G_{1,2}'$ is used to keep the consistency of $H_1$'s transcript. In \textbf{Update}, we chooses a random value for $UT_c(n)$ and stores it in table $\texttt{UT}$, and programed it to the output of $(K_n, ST_{c})$ in \textbf{Search}.
	
	Since $H_1$'s outputs in $G_{1,2}'$ and $G_{1,1}$ are perfectly indistinguishable, so we have
	$$\Pr[G_{1,1}=1]=\Pr[G_{1,2}'=1].$$
	
	 $G_{1,2}'$ and $G_{1,2}$ are also perfectly identical unless the $bad$ happens (set to $true$).
	 $\Pr[G_{1,2}'=1]-\Pr[G_{1,2}=1]\le \Pr[bad$ is set to true in $G_{1,2}']$
	 
	 Following \cite{Bos16}, the possibility for $H_1(K_n,ST_{c+1})$ already exists is the advantage of breaking the one-wayness of the trapdoor permutation which is $\verb"Adv"_{\Pi,\mathcal{B}_2}^{\verb"OW"}(1^{\lambda})$. Assume the query make $N$ queries, then we have 
	 
	 $$\Pr[G_{1,1}=1]-\Pr[G_{1,2}=1]=\Pr[G_{1,2}'=1]-\Pr[G_{1,2}=1]\le N \cdot \verb"Adv"_{\Pi,\mathcal{B}_2}^{\verb"OW"}(1^{\lambda}).$$
	
	%It is hard to distinguish these games due to the one-wayness of the trapdoor permutation. Then, we can conclude that there exists an efficiently $\mathcal{B}_2$ such that $$\Pr[G_{1,1}=1]-\Pr[G_{1,3}=1]\le 2N \cdot \verb"Adv"_{\Pi,\mathcal{B}_2}^{\verb"OW"}(1^{\lambda}) + 2Nq/2^{\lambda}$$ where $N$ is the number of times that queried $H_1$ and $H_2$.
	
	\textbf{Game} $G_{1,3}$: Similar to $G_{1,2}$, $G_{1,3}$ programs $H_2$. The same steps can be reused, giving that there is an adversary $B_3$, such that
	$$\Pr[G_{1,2}=1]-\Pr[G_{1,3}=1]\le N \cdot \verb"Adv"_{\Pi,\mathcal{B}_3}^{\verb"OW"}(1^{\lambda}).$$
	Note that, we can consider $B_2=B_3$ without loss of generality.

	\textbf{Game} $G_{1,4}$: In $G_{1,4}$, we keep the records of the random generated encrypted strings of the $H_1$ and $H_2$. In \textbf{Update}, we choose random values for update tokens and ciphertexts in Table \texttt{UT} and \texttt{e}, respectively. Then we program them identically to the outputs of the corresponding hash functions in \textbf{Search}. Then $G_{1,4}$ is exactly same as the $G_{1,3}$. More formally, $$\Pr[G_{1,4}=1]=\Pr[G_{1,3}=1]$$
	
	\begin{algorithm}[!htb]
		\caption{\textbf{Simulator} $\mathcal{S}_1$}\label{alg:aS}
		\begin{multicols}{2}
			\underline{$\mathcal{S}.$\textbf{Setup}($1^{\lambda})$}
			\begin{algorithmic}[1]
				\State $(\texttt{TSK}, \texttt{TPK})\leftarrow \texttt{TKeyGen}(1^{\lambda})$
				\State \textbf{N}, \textbf{T} $\leftarrow$ empty map
				\State $t=0$
				\State \Return $(\texttt{TPK}, \texttt{TSK}, \textbf{T}, \textbf{N})$
			\end{algorithmic}
			
			\underline{$\mathcal{S}.$\textbf{Update}()}\\
			\textit{Client:}
			\begin{algorithmic}[1]
				\State $\texttt{UT}[t]\leftarrow \{0,1\}^x$
				\State $\texttt{e}[t]\leftarrow \{0,1\}^y$
				\State Send $(\texttt{UT}[t],\texttt{e}[t])$ to the server.
				\State $t\leftarrow t+1$
			\end{algorithmic}
			
			\underline{$\mathcal{S}.$\textbf{Search}($\texttt{sp}(n), \texttt{Hist}(n), \texttt{cp}(n)$)}\\
			\textit{Client:}
			\begin{algorithmic}[1]
				\State $\hat{n}\leftarrow$ min $\texttt{sp}(n)$
				\State $K_{\hat{n}}\leftarrow \texttt{Key}[\hat{n}]$
				\State Parse $\texttt{cp}(n)$ as $t'$
				\If{$t'\ne\perp$} 
				\State Get the $c$-th search token $ST_c$ of previously queried keyword at time $t'$.
				\Else   
				\State Parse $\texttt{Hist}(n)$ as $((t_0,add,ind_0),$ $\cdots,(t_c,add,ind_c))$ 
				\If{$\texttt{Hist}(n)=\perp$}  
				\State \Return $\varnothing$
				\EndIf
				\For{$i=0$ to $c$}
				\State Set $H_1(K_{\hat{n}},ST_i)\leftarrow \texttt{UT}[t_i]$
				\State Set $H_2(K_{\hat{n}},ST_i)\leftarrow \texttt{e}[t_i]\oplus ind_i$
				\State $ST_{i+1}\leftarrow \Pi_{\texttt{TSK}}^{-1}(ST_i)$
				\EndFor
				\EndIf
				\State Send $(K_{\hat{n}}, ST_c)$ to the server.
			\end{algorithmic}
		\end{multicols}
	\end{algorithm}
	
	\textbf{Simulator} $\mathcal{S}_1$ With the contain pattern $\texttt{cp}$, the simulator can reuse the certain update token $UT$ to simulate the inclusion relationship between the keywords. We can use the search pattern $\hat{n}\leftarrow$ min $\texttt{sp}(n)$ and history $\texttt{Hist}$ to simulate the \textbf{Search} and \textbf{Update}. Similar to \cite{Bos16}, we have
	$$\Pr[G_{1,4}=1]=\Pr[\verb"DSSEIDEAL"_{\mathcal{A},\mathcal{S}_1}^{\Gamma_A}(1^{\lambda})=1]$$
	Finally, $$\Pr[\verb"DSSEREAL"_{\mathcal{A}}^{\Gamma_A}(1^{\lambda})=1]-\Pr[\verb"DSSEIDEAL"_{\mathcal{A},\mathcal{S}_1}^{\Gamma_A}(1^{\lambda})=1]$$$$\le \verb"Adv"_{F,\mathcal{B}_1}^{\texttt{prf}}(1^{\lambda})+2N \cdot \verb"Adv"_{\Pi,\mathcal{B}_2}^{\texttt{OW}}(1^{\lambda})$$ which completes the proof.\hfill{$\Box$}
\end{proof}

\noindent\textbf{Theorem \ref{th:bB}.} \emph{Let $F$ be a PRF, and $\Sigma$ be a IND-CPA secure Paillier cryptosystem. $\mathcal{L}_{\Gamma_B}=(\mathcal{L}_{\Gamma_B}^{Srch}$, $\mathcal{L}_{\Gamma_B}^{Updt})$, where $\mathcal{L}_{\Gamma_B}^{Srch}(n)=(\texttt{sp}(n))$, $\mathcal{L}_{\Gamma_B}^{Updt}(op,n,ind)=(n)$. Construction B is $\mathcal{L}_{\Gamma_B}$-adaptively backward-private.}

\begin{proof}
	For Theorem \ref{th:bB}, we also set a serial of games from $\verb"DSSEREAL"_{\mathcal{A}}^{\Gamma_B}(1^{\lambda})$ to $\verb"DSSEIDEAL"_{\mathcal{A},\mathcal{S}_2}^{\Gamma_B}(1^{\lambda})$.
	
	\textbf{Game} $G_{2,0}$: $G_{2,0}$ is exactly same as the real world game $\verb"DSSEREAL"_{\mathcal{A}}^{\Gamma_B}(1^{\lambda})$.$$\Pr[\verb"DSSEREAL"_{\mathcal{A}}^{\Gamma_B}(1^{\lambda})=1]=\Pr[G_{2,0}=1]$$
	
	\textbf{Game} $G_{2,1}$: Instead of calling $F$ when generating $UT_n$, $G_{2,1}$ picks a new random key when it inputs a new keyword $n$, and stores it in a table $\texttt{Key}$ so it can be reused next time. If an adversary $\mathcal{A}$ is able to distinguish between $G_{2,0}$ and $G_{2,1}$, we can then build a reduction able to distinguish between $F$ and a truly random function. More formally, there exists an efficient adversary $\mathcal{B}_1$ such that $$\Pr[G_{2,0}=1]-\Pr[G_{2,1}=1]\le \verb"Adv"_{F,\mathcal{B}_1}^{\texttt{prf}}(1^{\lambda}).$$
	
	\textbf{Game} $G_{2,2}$: We replace the bit string $bs$ with a all $0$ bit string. If an adversary $\mathcal{A}$ is able to distinguish between $G_{2,1}$ and $G_{2,2}$, we can then build an adversary $\mathcal{B}_2$ to break the semantic security of Paillier cryptosystem. More formally, there exists an efficient adversary $\mathcal{B}_2$ such that $$\Pr[G_{2,1}=1]-\Pr[G_{2,2}=1]\le \verb"Adv"_{\Sigma,\mathcal{B}_2}^{\texttt{IND-CPA}}(1^{\lambda}).$$
	
	\begin{algorithm}[!htb]
		\caption{\textbf{Simulator} $\mathcal{S}_2$}\label{alg:bS}
		\begin{multicols}{2}
			\underline{$\mathcal{S}.$\textbf{Setup}($1^{\lambda})$}
			\begin{algorithmic}[1]
				\State $(\texttt{SK}, \texttt{PK})\leftarrow \texttt{KeyGen}(1^{\lambda})$
				\State \Return $(\texttt{PK}, \texttt{SK},t)$
			\end{algorithmic}
			
			\underline{$\mathcal{S}.$\textbf{Update}($n$)}\\
			\textit{Client:}
			\begin{algorithmic}[1]
				\State $UT_{n}\leftarrow \texttt{Key}(n)$
				\State $e\leftarrow \texttt{Enc}(\texttt{PK}, 0\cdots0)$
				\State Send $(UT_{n},e)$ to the server.
			\end{algorithmic}
			
			\underline{$\mathcal{S}.$\textbf{Search}($\texttt{sp}(n)$)}\\
			\textit{Client:}
			\begin{algorithmic}[1]
				\State $\hat{n}\leftarrow$ min $\texttt{sp}(n)$
				\State $UT_{\hat{n}}\leftarrow Key(\hat{n})$
				%\State Parse $\texttt{HistUp}(n)$ as $(t_0,\cdots,t_c)$
				\State Send $UT_{\hat{n}}$ to the server.
			\end{algorithmic}
			
		\end{multicols}
	\end{algorithm}
	
	\textbf{Simulator} Now, we can simulator the \verb"DSSEIDEAL" with the leakage functions defined in this Theorem. We removed the useless part which will not influence the client's transcript. See Algorithm \ref{alg:bS} for more details. This two games is indistinguishable. So we have
	 $$\Pr[G_{2,2}=1]=\Pr[\verb"DSSEIDEAL"_{\mathcal{A},\mathcal{S}_2}^{\Gamma_B}(1^{\lambda})=1]$$
	
	Finally, $$\Pr[\verb"DSSEREAL"_{\mathcal{A}}^{\Gamma_B}(1^{\lambda})=1]-\Pr[\verb"DSSEIDEAL"_{\mathcal{A},\mathcal{S}_2}^{\Gamma_B}(1^{\lambda})=1]$$$$\le \verb"Adv"_{F,\mathcal{B}_1}^{\texttt{prf}}(1^{\lambda})+\verb"Adv"_{\Sigma,\mathcal{B}_2}^{\verb"IND-CPA"}(1^{\lambda})$$ which completes the proof.\hfill{$\Box$}
	
\end{proof}

\end{document}